\documentclass[aps,prapplied,amsmath,amssymb,twocolumn,superscriptaddress,longbibliography]{revtex4-2}

\usepackage[dvipdfmx]{graphicx}
\usepackage{latexsym}
\usepackage{dcolumn}
\usepackage{amsthm}
\usepackage{bm}
\usepackage{bibentry} 
\usepackage{siunitx}
\usepackage{mathtools}
\usepackage{enumitem}  
\usepackage{algorithm}
\usepackage[noend]{algpseudocode}
\mathtoolsset{showonlyrefs=true}

\theoremstyle{definition}
\newtheorem{theorem}{Theorem}

\newtheorem{definition}{Definition}

\newtheorem{lemma}{Lemma}

\algrenewcommand\algorithmicrequire{\textbf{Input:}}
\algrenewcommand\algorithmicensure{\textbf{Output:}}
\algrenewcommand\algorithmicforall{\textbf{for each}}

\renewcommand{\bm}[1]{{\boldsymbol{\mathrm{#1}}}}
\DeclareMathOperator*{\argmax}{arg\,max}

\begin{document}

\preprint{APS/123-QED}

\title{Probe position determination with multichannel \textit{I--V} measurements in a two-dimensional sheet: Computational method and mathematical analysis}%

\author{Tomoyuki Iori}
\email{tomoyuki.iori@ieee.org}  
\affiliation{Graduate School of Information Science and Technology, Osaka University, Suita, Osaka 565-0871, Japan}
\affiliation{Space Tracking and Commmunications Center, \\Japan Aerospace Exploration Agency, Tsukuba, Ibaraki 305-8505, Japan}

\author{Ryu Yukawa}
\email{r.yukawa@tohoku.ac.jp}  
\affiliation{International Center for Synchrotron Radiation Innovation Smart (SRIS), Tohoku University, Sendai, Miyagi 980-8572, Japan}

\date{\today}


\begin{abstract} 

Atomically thin films and surfaces exhibit many distinctive two-dimensional electronic properties that are absent in bulk crystals. \textit{In situ} microscale multi-probe measurements have been utilized as an effective method to identify the electrical conductivity of such thin films and surfaces. Precise determination of multi-probe positions is crucial for accurate characterization of the conductance. However, traditional methods that use microscopes for determining multi-probe positions often impose significant constraints on experimental setups. In some cases, installing a microscope is not even feasible. Therefore, in this study, we propose a novel method to determine probe positions using electrical signals from the probes. This method enables precise determination of probe positions using a reference sheet and reference probes, even at low or high temperatures and under ultra-high vacuum or high-pressure conditions. The proposed method simplifies the integration of microscale multi-probe measurement systems into various devices, thereby advancing research on thin films and surfaces.

\end{abstract}

\maketitle

\section{INTRODUCTION}

Atomically thin films and surfaces have revealed various electronic transport properties that do not appear in bulk crystals; for example, ultra-thin films have demonstrated high carrier mobilities~\cite{Novoselov2004,Zhang2005,Schwierz2015}, superconductive transitions~\cite{Wang2012,Ichinokura2016,Tsen2016}, and metal--insulator transitions~\cite{Tanikawa2004,Radisavljevic2013,Yukawa2021}, whereas surfaces have revealed spin currents~\cite{Hsieh2008}, 2D electron gas on transparent substrates~\cite{Meevasana2011,Santander-Syro2011,Moser2013}, and topological superconductivity~\cite{Zhang2018}.
Despite these interesting properties, accurate measurement of the electrical conductivity of atomically thin films and crystal surfaces remains challenging because they are easily affected by the atmosphere or residual gases. Therefore, it is desired that samples are directly transferred from a sample growth/treatment chamber to a measurement chamber without exposure to the air.

For the accurate measurements of low-resistivities of thin films/surfaces, the current-voltage (\textit{I--V}) measurement using four probes is indispensable. 
In the four-probe method, the electronic current is supplied through two probes while the voltage change is measured across the remaining two. Because no electronic current flows between the voltage-measuring probe and the sheet, no voltage drops originating from the contact resistance are induced. Thus, accurate measurements that are not affected by contact resistance are realized using the four-probe method~\cite{Wenner1916,Valdes1954}. 
In addition, surface-sensitive measurements with probe spacing reduced to the order of micrometers~\cite{Shiraki2001,Miccoli2015} are necessary for the precise measurement of thin films/surfaces on substrates/bulk with a non-negligible electronic conductivity.
Hence, the adoption of four probes arranged at microscale intervals is essential for measuring the electronic conductivity in thin films and crystal surfaces~\cite{Shiraki2001,Miccoli2015}. 

Employing a larger number of probes ($N\geq 4$) and arranging them in a non-linear configuration allow for the acquisition of important properties related to the electronic characteristics. By placing multi-probes at varying probe intervals, it becomes possible to perform precise measurements that distinguish between thin-film/surface conductivity and substrate/bulk conductivity~\cite{Shiraki2001,Durand2016,Ko2018}. In addition, setting the probes in a non-linear configuration enables the acquisition of the anisotropic transport properties~\cite{Kanagawa2003,Uetake2012,Edler2015,Ma2017}. The use of a non-linear configuration under a magnetic field enables even Hall measurements, which provide important information about carriers~\cite{Zhang2005,Wang2014}. Therefore, \textit{in situ} measurements using micro-multi probes with non-linear configurations have significant advantages over the conventional four-probe methods for investigating the electronic properties of thin films or surfaces and are one of the technologies that attract growing attention.

To obtain accurate transport properties from such multi-probe measurements, probe positions must be known. However, determination or verification of the probe positions is often challenging. In a typical system, the measurement stage installed inside an ultra-high vacuum (UHV) chamber is located away from the viewports, thereby limiting the magnification of the view near the probes when using an optical microscope. Although a high spatial resolution is expected using a scanning electron microscope (SEM)~\cite{Shiraki2001,Uetake2012,Edler2015,Durand2016,Kim2007}, installing the SEM system inside the UHV equipment is costly and limits the arrangement of the apparatus. In addition, direct determination using the microscopes becomes considerably more difficult when the probes are hidden by a probe holder (Fig.~\ref{fig:Intro}) or when the measurement stage is covered by a radiation shield for low-temperature measurements. To avoid these difficulties, a method to determine the probe positions based on a completely different principle than the conventional methods using microscopes is desired.

\begin{figure} 
    \centering
    \includegraphics[width=9cm]{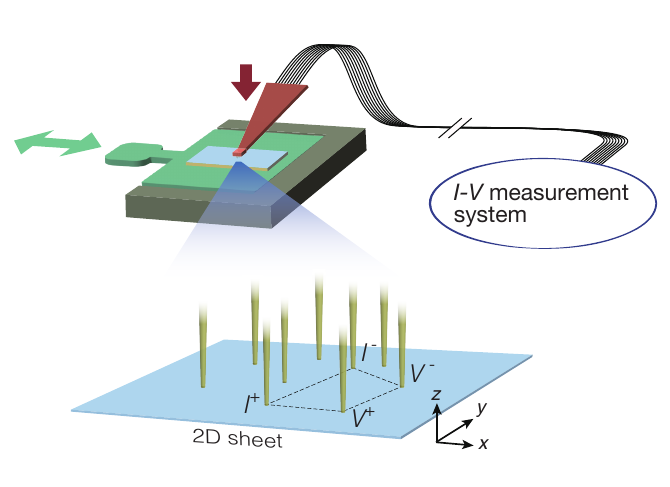}
    \caption{Schematic of an experimental setup for the multi-probes measurement on a two-dimensional sheet. After transferring a sample folder onto a stage, a tip with multi-probes is approached to the sample surface. The sample transfer and the contact of the probes to the surface are performed in a vacuum chamber. \textit{I--V} measurements are conducted via electronic cables, each connecting a probe to the \textit{I--V} measurement system placed outside of the vacuum chamber. \textit{I--V} measurements are conducted by switching various sets of four probes.}
    \label{fig:Intro}
\end{figure}
    
In this study, we propose a method to precisely determine the positions of multi-probes by analyzing the electrical signals from the probes. This approach uses a reference sheet and reference probes instead of a microscope. First, a reference sheet with a uniform resistivity is transported to the measurement stage, and both target probes and reference probes are placed on the sheet. Subsequently, \textit{I--V} measurements are repeatedly conducted with various sets of four probes, where each set of probes is individually selected for each measurement. From the obtained data, the positions of the target probes relative to the reference probes are determined. 
We have shown that, under specific conditions, both probe positions and sheet resistances can be simultaneously determined, without needing to know the resistance of the reference sheet, which is susceptible to external factors.

\paragraph*{Notation:}
Let $\mathbb{R}$ denote the set of all real numbers. 
For two positive integers $n, m$, $\mathbb{R}^n$ and $\mathbb{R}^{n \times m}$ denote the set of all $n$-dimensional vectors and the set of all $n \times m$ matrices with real components, respectively.
All the vectors are denoted in bold letters, for example, $\bm{p} \in \mathbb{R}^2$.
The symbol $\| (\cdot) \|$ denotes the Euclidean norm of a vector, for example, $\| \bm{p} \| = \sqrt{p_1^2 + p_2^2}$ for $\bm{p} = [p_1, p_2]^\top \in \mathbb{R}^2$.
The symbol~$\hat{(\cdot)}$ denotes the estimate of an unknown parameter, for example, $\hat{\bm{p}} \in \mathbb{R}^2$ is the estimate of $\bm{p}$.
For a positive integer $a$, $a!$ denotes the factorial of $a$, that is, $a! \coloneqq 1 \times 2 \times \cdots \times (a-1) \times a$.
For two positive integers $a, b$ with $a > b$, $\tbinom{a}{b}$ denotes the number of combinations of $b$ elements chosen from $a$ elements, that is, $\tbinom{a}{b} \coloneqq a! / (b! (a-b)!)$.

\section{MATHEMATICAL ANALYSIS \label{sec:Math}}

In this section, we provide computational methods and mathematical formulations, theories, and proofs, for determining probe positions on a two-dimensional reference sheet with \textit{I--V} measurements.

\subsection{Probe-Position Determination as a Nonlinear Optimization Problem}
We consider a two-dimensional reference sheet with uniform isotropic sheet resistance $\rho_{\rm{s}}$.
Let $N$ denote the number of probes placed on the sheet.
Some of the probes are assumed to be reference probes, whose positions are known.
The resistance obtained from \textit{I--V} measurements depends only on distances between each pair of probes used in the measurements.
In other words, the resistance $R^{ab}_{cd}$ with the $I^+$-probe $a$, $V^+$-probe $b$, $V^-$-probe $c$, and $I^-$-probe $d$ is given by~\cite{Miccoli2015}
\begin{equation}
  R^{ab}_{cd} =  \frac{\rho_{\mathrm{s}}}{2 \pi}\ln\left(	\frac{s_{bd}s_{ac}}{s_{ab}s_{cd}} \right), \label{eq:R}
\end{equation}
where $s_{ab} = \| \bm{p}_a - \bm{p}_b\|$ denotes the distance between probes $a$ and $b$ with their positions $\bm{p}_a \in \mathbb{R}^2$ and $\bm{p}_b \in \mathbb{R}^2$, respectively.
The resistance $R^{ab}_{cd}$ is experimentally determined by the relation $R^{ab}_{cd}=V/I$, where $V$ is the voltage drop from probe $b$ to probe $c$, and $I$ is the current flow from probe $a$ to probe $d$. 
The sheet resistance $\rho_{\rm{s}}$ can be uniquely determined from~\eqref{eq:R} if the positions of probes $a$, $b$, $c$, and $d$ are known and the resistance $R^{ab}_{cd}$ is measured.

Definition~\eqref{eq:R} can be rewritten as 
\begin{equation}
    C^{ab}_{cd}(\rho_{\mathrm{s}})s_{ab}s_{cd} = s_{bd}s_{ac}, \label{eq:R2}
\end{equation}
where $C^{ab}_{cd}(\rho_{\mathrm{s}}) \coloneqq \exp(2 \pi R^{ab}_{cd} / \rho_{\mathrm{s}})$. 
By substituting $s_{ab} = \|\bm{p}_a - \bm{p}_b\|$ into~\eqref{eq:R2} and subtracting the right-hand side from the left-hand side, we obtain
\begin{equation}
    C^{ab}_{cd}(\rho_{\mathrm{s}})\|\bm{p}_a - \bm{p}_b\|\|\bm{p}_c - \bm{p}_d\| - \|\bm{p}_b - \bm{p}_d\|\|\bm{p}_a - \bm{p}_c\|,
\end{equation}
which must be zero if $\rho_\mathrm{s}$, $\bm{p}_a$, $\bm{p}_b$, $\bm{p}_c$, and $\bm{p}_d$ are the true values. 
Hence, the following function of the estimated probe positions and sheet resistance serves as the error of the estimates: 
\begin{multline}
    \mathrm{Err}(\hat{\bm{p}}_a, \hat{\bm{p}}_b, \hat{\bm{p}}_c, \hat{\bm{p}}_d, \hat{\rho}_{\mathrm{s}}) \\
    \coloneqq C^{ab}_{cd}(\hat{\rho}_{\mathrm{s}})\|\hat{\bm{p}}_a - \hat{\bm{p}}_b\|\|\hat{\bm{p}}_c - \hat{\bm{p}}_d\| - \|\hat{\bm{p}}_b - \hat{\bm{p}}_d\|\|\hat{\bm{p}}_a - \hat{\bm{p}}_c\|, \label{eq:E}
\end{multline}
which becomes zero when the estimates $\hat{\bm{p}}_a$, $\hat{\bm{p}}_b$, $\hat{\bm{p}}_c$, $\hat{\bm{p}}_d$, and $\hat{\rho}_\mathrm{s}$ coincide with the true values.

For the set of all probe quadruples $\mathcal{M}$ with measured resistances, we can compute the estimates $\hat{\bm{P}} = (\hat{\bm{p}}_1, \hat{\bm{p}}_2, \dots, \hat{\bm{p}}_N)$ and $\hat{\rho}_\mathrm{s}$ by minimizing the following cost function:
\begin{equation}
    \sum_{(a,b,c,d) \in \mathcal{M}} w^{ab}_{cd}\mathrm{Err}(\hat{\bm{p}}_a, \hat{\bm{p}}_b, \hat{\bm{p}}_c, \hat{\bm{p}}_d, \hat{\rho}_\mathrm{s})^2, \label{eq:opt}
\end{equation}  
where the summation is taken over all quadruples in $\mathcal{M}$, and $w^{ab}_{cd} \in \mathbb{R}$ denotes the weight for each quadruple.
The cost function~\eqref{eq:opt} is nonlinear and nonconvex, and it can have multiple local minima. 
Therefore, we need to clarify how many and which resistance measurements are required to determine the probe positions. 

\subsection{Theoretical Results on Relation between Probe Positions and Resistance Measurements}
For the four probes $a$, $b$, $c$, and $d$, there are 24 patterns of measurement resistance: $R^{ab}_{cd}, R^{bc}_{da}, R^{cd}_{ab}, \dots$, which are obtained by permuting the roles of the probes.
However, these patterns do not provide independent information on the positions of probes $a$, $b$, $c$, and $d$. 
For example, when the resistance $R^{ab}_{cd}$ is measured, the measurement of another resistance $R^{ac}_{bd}$ does not provide any new information on the probe positions because the swap of $V^+$ and $V^-$ probes only reverses the sign of the resistance, that is, $R^{ab}_{cd} = -R^{ac}_{bd}$.
Specifically, the following theorem holds.
\begin{theorem} \label{thm:indepMeas}
    All 24 resistance measurement patterns can be computed from only two of them.
\end{theorem}

Moreover, \eqref{eq:R} indicates that the resistance is not changed by scaling distances between probes; that is, the resistance $R^{ab}_{cd}$ is equal to the resistance $R^{a'b'}_{c'd'}$ measured with probe positions $\bm{p}_{a'}$, $\bm{p}_{b'}$, $\bm{p}_{c'}$, and $\bm{p}_{d'}$ defined by $\bm{p}_{a'} = \lambda \bm{p}_{a}$ for any $\lambda > 0$. 
This implies that another equation satisfied by the probe positions is required to determine the probe positions.
The key observation is that~\eqref{eq:R} depends only on distances between probes, and hence, it fails to consider the fact that all probes are placed on a single two-dimensional plane, that is, the surface of the reference sheet.
This fact can be formulated using the well-known \emph{cosine rules}: 
\begin{equation}
    s_{ab}^2 + s_{ac}^2 - s_{bc}^2 = 2 \bm{v}_{ab}^\top \bm{v}_{ac}, \label{eq:cos_rule}
\end{equation}
where $\bm{v}_{ab} = \bm{p}_b - \bm{p}_a$ and $\bm{v}_{ac} = \bm{p}_c - \bm{p}_a$. 
Using~\eqref{eq:cos_rule}, the following theoretical results were obtained:

\begin{lemma} \label{thm:3probes}
  Let $a$, $b$, and $c$ denote reference probes, and let $d$ denote another probe whose position is unknown. 
  Suppose that probes $a$, $b$, and $c$ are not aligned in a straight line, resistances $R^{da}_{bc}$ and $R^{da}_{cb}$ are measured, and sheet resistance $\rho_{\rm{s}}$ is known. 
  Then, the number of possible positions of probe $d$ is at most two.
\end{lemma}

\begin{lemma} \label{thm:4probes}
    Let $a$, $b$, $c$, and $d$ denote reference probes, and let $e$ denote another probe whose position is unknown.
    Suppose that the three probes $a$, $b$, and $c$ are not aligned in a straight line, that the same is true for $b$, $c$, and $d$, and that resistances $R^{ea}_{bc}$, $R^{ea}_{cd}$, and $R^{ea}_{db}$ are measured.
    Then, the sheet resistance $\rho_{\rm{s}}$ and the position of probe $e$ are uniquely determined, except in cases where probes $a$, $b$, $c$, and $d$ satisfy a certain condition, which is given by~\eqref{eq:4probes_condition} in Appendix~\ref{appsec:4probes}. 
\end{lemma}

\begin{lemma} \label{thm:5probes}
    Let $a$, $b$, and $c$ denote reference probes, and let $d$ and $e$ denote any other two probes whose positions are unknown.
    Suppose that probes $a$, $b$, and $c$ are not aligned in a straight line, and resistances $R^{da}_{bc}$, $R^{da}_{cb}$, $R^{ea}_{bc}$, $R^{ea}_{cb}$, and $R^{bd}_{ce}$ are measured. 
    Then, the number of possible values of the sheet resistance $\rho_{\rm{s}}$, as well as the number of possible positions of probes $d$ and $e$ that realize the measured resistance, is finite.
    For each candidate of the sheet resistance $\rho_{\rm{s}}$, the positions of $d$ and $e$ have at most two candidates.
\end{lemma}

Appendix~\ref{app:proofs} provides the proofs of Lemmas~\ref{thm:3probes},~\ref{thm:4probes}, and~\ref{thm:5probes}.
These lemmas can be used to determine the minimum number of measurements required to reduce the unknown sheet resistance and probe positions to a finite number of candidates.

\begin{theorem} \label{thm:numMeas}
    Let $N$ denote the number of probes. 
    Suppose that the positions of the three reference probes are known and that any three of all probes are not aligned in a straight line. 
    Then, $2 + 3(N-4)$ resistance measurements are sufficient to determine two possible patterns of the unknown probe positions if the sheet resistance is known.
    Futhermore, even if the sheet resistance is unknown, the number of candidates for the sheet resistance and possible positions of the unknown probes is finite.
\end{theorem}
\begin{proof}
    Suppose the sheet resistance is known.
    For the reference probes, Lemma~\ref{thm:3probes} requires two measurements to determine at most two candidates for the fourth probe.
    Once the position of the fourth probe is fixed to one of the two candidates, Lemma~\ref{thm:4probes} requires three measurements to determine the position of the fifth probe uniquely.
    This can be repeated until the positions of all unknown probes are determined.
    The required measurements depend on the new probe at each step, and therefore, no measurements are redundant.
    Hence, the total number of measurements is $2 + 3(N-4)$.

    When the sheet resistance is unknown, Lemma~\ref{thm:5probes} ensures that the numbers of candidates for the sheet resistance and the possible positions of the fourth and fifth probes are finite, requiring five measurements.
    Once a candidate for the sheet resistance and positions of the fourth and fifth probes are fixed, Lemma~\ref{thm:4probes} requires three additional measurements for each unknown probe.
    Hence, the total number of measurements is $5 + 3(N-5) = 2 + 3(N-4)$, completing the proof.
\end{proof}

According to Theorem 2, the possible positions of the target probes are determined as two patterns from \textit{I--V} measurements when the number of reference probes is three and the sheet resistance is known. 
However, in a practical experimental system, there is a large possibility that one can choose the true pattern from two candidates because they typically differ clearly from each other. 

Figure~\ref{fig:2pattern} shows an example of the probe positions simulated for $N=8$, where probes 1--3 are reference probes. 
The two probe positions, A and B, provide identical resistance measurements for any quadruplet of probes, which can be verified using~\eqref{eq:R}. 
As shown in the figure, the positions of the target probes are completely different for each pattern. 
Therefore, in an actual experimental system, one can select the true pattern from two candidates by considering prior knowledge of the probe configurations, including the visually obtained rough layout of the probe positions.
\begin{figure}
    \centering
    \includegraphics[width=4.5cm]{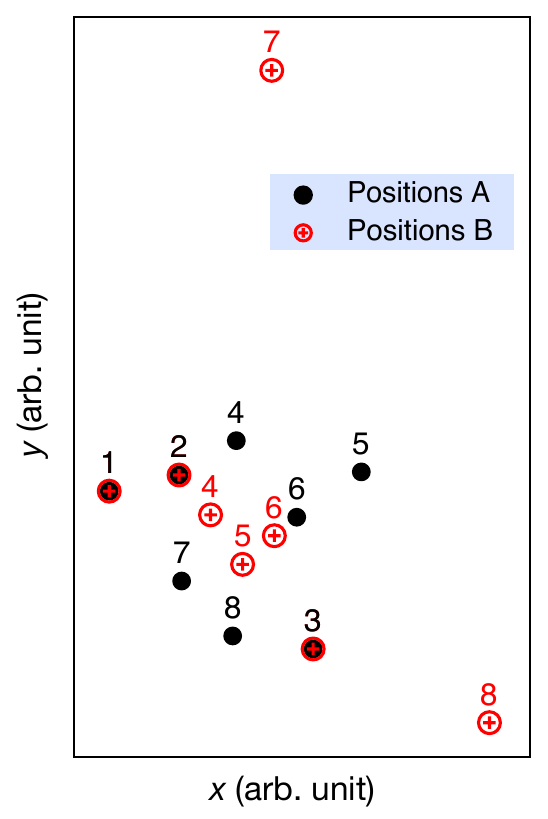}
    \caption{Example of the eight probe positions. We set probes 1--3 as reference probes. Two patterns are possible for the positions of probes 4--8: Positions A (black circles) and B (red circles). \textit{I--V} measurements with these two patterns yield the same results.}
    \label{fig:2pattern}
\end{figure}

Finally, the entire optimization procedure is summarized in Algorithm~\ref{alg:opt}, which is implemented using \texttt{Python} with the \texttt{SciPy} library and its optimization routine.
\begin{figure}[t]
    \begin{algorithm}[H]
        \centering
        \caption{Determination of Probe Positions\label{alg:opt}}
        \begin{algorithmic}[1]
            \Require{resistance measurements $R^{ab}_{cd}$ and $R^{ba}_{cd}$ for all quadruples of probes $(a, b, c, d)$, weight $w^{ab}_{cd}$ for each measurement, true positions of reference probes $\bm{p}_1$, $\bm{p}_2$, $\bm{p}_3$, and tolerance $\epsilon$ for minimization}
            \Ensure{estimates of the target probes' positions $\hat{\bm{p}}_i\ (i=4,5,\dots,N)$ and the sheet resistance $\hat{\rho}_{\mathrm{s}}$}
            \State{Set initial guesses of estimates $\hat{\bm{p}}_i$ and $\hat{\rho}_{\mathrm{s}}$\label{algline:init}}
            \State{Minimize cost function~\eqref{eq:opt} defined by $R^{ab}_{cd}$, $R^{ba}_{cd}$, and $w^{ab}_{cd}$ with respect to $\hat{\bm{p}}_i\ (i = 4, 5, \dots, N)$ and $\hat{\rho}_{\mathrm{s}}$ while fixing $\hat{\bm{p}}_i = \bm{p}_i\ (i = 1, 2, 3)$}
            \If{the cost becomes less than $\epsilon$}
                \State\Return{the estimates $\hat{\bm{p}}_i$ and $\hat{\rho}_{\mathrm{s}}$}
            \Else
                \State{Go back to Line~\ref{algline:init} and select different initial guesses}
            \EndIf
        \end{algorithmic}
    \end{algorithm}
\end{figure}

\section{EXPERIMENTAL DEMONSTRATION AND DISCUSSION}\label{sec:Exp}

\subsection{Conditions for practical applications}\label{subsec:conditions}

To adapt the method proposed in Section~\ref{sec:Math} to an actual experimental system, the following conditions are required for probes and a reference sheet:

\begin{enumerate}[label=(\roman*)]
    \item The reference sheet should be sufficiently large against the probe spacings so that the sheet can be regarded as an infinitely large sheet. \label{enum:size}
    
    \item The thickness of the reference sheet should be sufficiently thin so that the sheet can be regarded as a 2D sheet. \label{enum:thickness}
    
    \item The sheet should follow Ohm’s law, and sheet resistance should be sufficiently isotropic. In addition, the sheet resistance should be sufficiently uniform, so that resistances satisfy the following relations: $R^{ab}_{cd} = R^{ba}_{dc}$ and $R^{ab}_{cd} = R^{ba}_{cd}+R^{cb}_{ad}$.
    Temporal invariance of the sheet resistance during the \textit{I--V} measurement is also required. \label{enum:ohm}
    
    \item The contact areas between the probe and the sheet should be sufficiently small so that the areas can be regarded as points in the 2D sheet. In addition, the contacts should be ohmic. \label{enum:contact}

\end{enumerate}
Details are as follows:

\ref{enum:size} When the size of the reference sheet is finite, measured resistance values deviate from the value computed from~\eqref{eq:R}, which assumes the infinitely large sheet. In a simplified model with a square array of four probes centered on a circular sheet of diameter $d$, the deviation of $R^{ab}_{cd}$ from the infinitely large sheet model is approximately 5$\%$ when $d/s=10$~\cite{Vaughan1961,Miccoli2015}, where $s$ is the probe spacing. The deviation from the infinitely large sheet model becomes smaller for a larger $d/s$ ratio and is less than 2$\%$ when $d/s=16$.

\ref{enum:thickness} When the reference sheet has a finite thickness, measured resistance values deviate from the value computed from~\eqref{eq:R}.
In a simplified model, where four probes are equally spaced in a straight line on an infinitely large sheet with thickness $t$, the deviation from the ideal two-dimensional sheet is approximately 5$\%$  when $s'/t=2$~\cite{Albers1985,Miccoli2015}, where $s'$ is the probe spacing. The deviation from the ideal 2D sheet decreases to less than 2$\%$ and 1$\%$ when $s'/t=3$ and $s'/t=5$, respectively. Although the calculation models in (i) and (ii) are highly simplistic for practical usage, they remain helpful for estimating the deviations from the ideal 2D sheet.

\ref{enum:ohm} The reference sheet should follow Ohm's law in the measurement current range, and any materials that generate non-Ohmic voltages, such as the Hall voltage in a magnetized material, should not be used.
Generally, microscopic variations in sheet resistance are inevitable due to the presence of surface roughness, oxidization, or impurities (Fig.~\ref{fig:MicroscopeImage}). Sometimes, damages to the sheet may occur from the probes' approach. Even with such microscopic variations in the resistance, the proposed method is applied if the sheet resistance is sufficiently uniform within the probe distance range.  The uniformity of the reference sheet is confirmed by examining whether relations 
$R^{ab}_{cd} = R^{ba}_{dc}$ and $R^{ab}_{cd} = R^{ba}_{cd}+R^{cb}_{ad}$ hold in the \textit{I--V} measurements. 
Using a 2D sheet with isotropic resistance is also required. Contrary to the uniform resistivity, the confirmation of the isotropic resistivity requires multiple probes of known configuration~\cite{Kanagawa2003}.Thus, isotropic resistance must be ensured in advance. 
The temporal invariance of the sheet resistance is confirmed by repeating \textit{I--V} measurements at both the beginning and end of a series of measurements with an identical set of probes. 
\begin{figure}
    \centering
    \includegraphics[width=9cm]{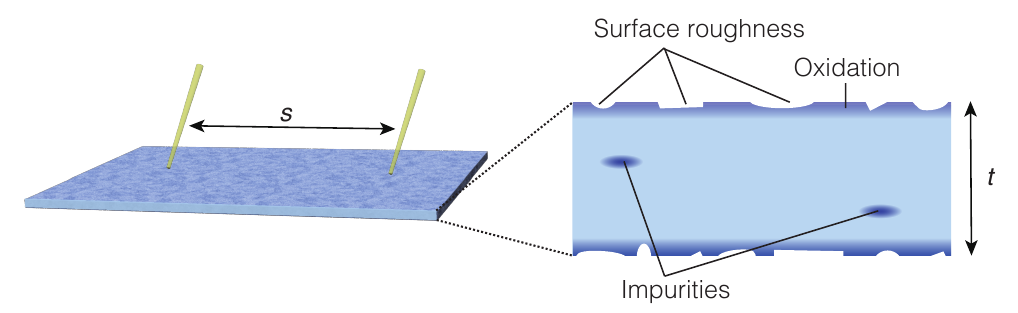}
    \caption{Microscopic image of a reference sheet. A reference sheet with surface roughness, oxidation, and impurities, in a microscopic view, can be adopted as the reference sheet when the sheet and the probes satisfy the conditions shown in Subsection~\ref{subsec:conditions}.}
    \label{fig:MicroscopeImage}
\end{figure}

\ref{enum:contact} The size of the contact area between the probes and the sheet should be much smaller than the probe distances. Particularly, the uncertainty of the reference probe positions directly affects the positioning accuracy of the target probes. In addition, ohmic contact is required between the probes and the sheet. Contrary to non-ohmic contacts, which include \textit{p--n} junctions and Schottky barriers, linear \textit{I--V} curves are obtained with the ohmic contact. The contact types depend on the materials of probes and sheets, and the ohmic contact is verified through linear curves of \textit{I--V} measurements.

Finally, some conditions, such as the isotropic resistivity of the reference sheet, cannot be verified within the proposed method and hence must be ensured in advance, that is, in the preparation process of the reference sheet. This is because the main application of the multi-probes, whose positions are determined by the proposed method, includes the measurements of anisotropic resistivity. We believe that preparing a reference sheet that meets the conditions in \ref{enum:size}--\ref{enum:contact} is more feasible than installing a precise measurement and control system of probe positions in the UHV chamber.

\subsection{Experimental demonstration}
The results of the theoretical analysis were experimentally demonstrated for the case with $N = 9$ as shown in Fig.~\ref{fig:ExpSetup}(a). Here, three probes with $n=1$ to 3 were used as the reference probes of known configuration. On the other hand, the positions of the six target probes ($n=4,5,\dots,9$) and the sheet resistance of the reference sheet were set as the parameters that should be determined by minimizing~\eqref{eq:opt}. To avoid external noise, a measurement stage and the probes were installed in a vacuum chamber. The probes were contacted with the film surface with the $z$-motion of the probe holder, similar to the schematic shown in Fig.~\ref{fig:Intro}. Each probe was connected to an \textit{I--V} measurement device using a switching system outside the vacuum chamber.
Contact of all the probes with the reference sheet was confirmed by monitoring the electrical resistance between each probe and the sheet.

A stainless steel (grade 304) film, which was \SI{10}{\um} in thickness, approximately \SI{13}{\mm}~$\times$~\SI{13}{\mm} in size, and placed on a glass substrate, was employed as a reference sheet.
Since the probe spacing was larger than \SI{80}{\um} for the current setup, the stainless steel film was adequately regarded as a two-dimensional sheet in the variation of less than 1$\%$. In addition, since probes with an average spacing of approximately \SI{610}{\um} were placed near the center of a reference sheet, the average variation from the infinity large sheet model was assumed to be less than 2$\%$.
Therefore, the sheet is adequately considered as an infinitely large two-dimensional sheet and the mathematical method presented in Section~\ref{sec:Math} is applicable.

First, to demonstrate the argument of Theorem~\ref{thm:indepMeas}, \textit{I--V} measurements were conducted for $\tbinom{9}{4} \times 6 = 756$ patterns of the probes, where multiplication by $6$ represents the six permutations for each quadruple of the probes. 
Although the number of permutations in each quadruple was $4! = 24$, we omitted trivial permutations, such as the exchange of $V^+$ and $V^-$ or $I^+$ and $I^-$, which led to $\frac{24}{2 \times 2} = 6$. 
According to Theorem~\ref{thm:indepMeas}, for a quadruple of probes $(a, b, c, d)$, the resistance measurements $R^{ab}_{cd}$ and $R^{ba}_{cd}$ are sufficient to determine the other four measurements. 
We therefore compared the four measurements with the predicted values computed from $R^{ab}_{cd}$ and $R^{ba}_{cd}$ theoretically.
For all the $\tbinom{9}{4} \times (6 - 2) = 504$ predictions, $90\%$ of them were within $4\%$ error of the measurements, which confirms the validity of Theorem~\ref{thm:indepMeas}.
Next, to determine the probe positions by solving~\eqref{eq:opt}, $\tbinom{9}{4} \times 2 = 252$ measurements were used, where the multiplication by $2$ corresponds to the independent measurements $R^{ab}_{cd}$ and $R^{ba}_{cd}$ for each quadruple. 

The current passing through the two probes varied from \SIrange{-10}{10}{\mA}, and the voltage difference was measured across the other two probes. Representative results of these measurements are presented in Fig.~\ref{fig:ExpSetup}(b). Linear \textit{I--V} curves were obtained for all combinations, indicating the formation of ohmic contacts between the tungsten probes and the sheet. Fluctuations in the acquired data points varied depending on the alignment of the measurement probes. Consequently, the reliability of the resistance values derived from linear fittings varied. 
Therefore, the residue in the linear fitting served as an index of the measurement reliability. 
Specifically, the square root of the residue, whose unit is \si{\ohm}, is divided by the absolute value of the resistance measurement for obtaining a dimensionless index. 
The weight $w^{ab}_{cd}$ in~\eqref{eq:opt} is then set as the reciprocal of the index to prioritize reliable measurements.

As a consequence of the nonlinear minimization, the target probe positions and the sheet resistance, $\rho_{\rm{s}}$, are obtained. As shown in Fig.~\ref{fig:ExpSetup}(c), the positions of the target probes are in excellent agreement with the actual probe positions seen in the microscope image. The resultant mean error is $\SI{7.6}{\um}$, corresponding to an error of approximately 1.2$\%$ against the mean probe spacings. 
In addition, $\rho_{\mathrm{s}}= \SI{.0798}{\ohm}/\square$\footnote{The symbol $\si{\ohm}/\square$ denotes the unit of sheet resistance.} is determined. 
By adopting the thickness of the reference sheet used in this demonstration ($t = \SI{10}{\um}$), a resistivity of \SI{.798e-6}{\ohm\m} was obtained. Although there could be a potential imprecision in the surface roughness, oxidization, or impurities (Fig.~\ref{fig:MicroscopeImage}), the obtained resistivity agreed well with the reported value for the stainless steel ($\rho = \SI{.713e-6}{\ohm\m}$ at room temperature~\cite{Ho}). 
These results reinforce the validity of the mathematical method, demonstrating that the sheet resistance as well as the positions of the probes can be determined from \textit{I--V} measurements with the three reference probes.

\begin{figure}
    \centering
    \includegraphics[width=9cm]{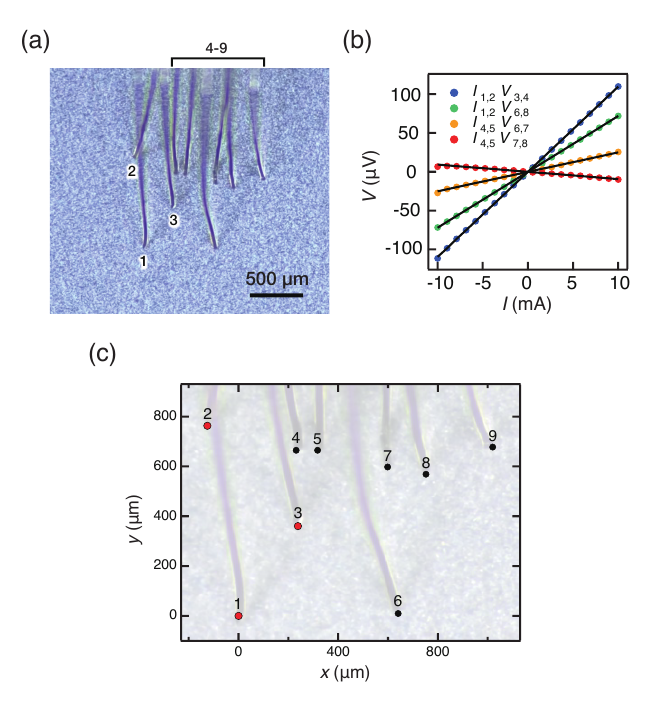}
    \caption{Experimental demonstration of the multi-probe determination. 
    (a) Microscope image of nine probes placed on the stainless-steel sheet. (b) Representative data of \textit{I--V} measurements. Here, the voltage offsets are adjusted so that the \textit{I--V} curves cross \SI{0}{\uV} at \SI{0}{\mA}. (c) The positions of the target probes (black circles) obtained from the mathematical analysis. Red circles represent the positions of the reference probes obtained from the microscope image.}
    \label{fig:ExpSetup}
\end{figure}

\subsection{Discussion}

In Section~\ref{sec:Math}, we proved that when the number of reference probes is three and $N\geq 5$, the number of possible sheet resistances is limited to be finite, and for each sheet resistance, two patterns of target probe arrangements exist. This implies that other minima can exist for~\ref{eq:opt}.
However, the demonstration in the previous section reveals that the equation reaches its minimum with only one $\rho_{\rm{s}}$ value and a single pattern of the target probes. In the minimization process, the same results were obtained even when the initial values for $\rho_{\rm{s}}$, which differed by several digits from the true value, were used.
This suggests that either no other minima exist for $\rho_{\rm{s}}$ or they are significantly different. The failure to converge to the second pattern of the target probe arrangements can be attributed to the exclusion of extreme cases, such as when the distance between probes is narrower than the diameter of the probes themselves. This demonstration indicates that using three reference probes in actual experiments allows for the unique determination of the $\rho_{\rm{s}}$ value and the arrangement of the target probes.

Given that the constraints discussed in Subsection~\ref{subsec:conditions} are met, the arrangements of target and reference probes are designed depending on the precision of the \textit{I--V} measurements. 
Two examples of probe arrangements are presented in Fig.~\ref{fig:PossibleSetup}. Here, the reference probes at fixed positions are placed away from or around the target probes.  By using a sheet sample with an uniform resistivity, it becomes possible to obtain the electronic transport properties of the sample while determining the overall positions of the target probes. 

The proposed method can be applied to any sheets that satisfy the conditions listed in Subsection~\ref{subsec:conditions}. This means that the materials for the reference sheet are not limited to metals.
Layers with isotropic conductive properties, such as graphite and two-dimensional metallic gas formed on an insulator, are also applicable as the reference sheet. On the other hand, the robustness of the reference sheet against repeated approaches of the probes becomes a crucial factor for long-term use. In this case, the use of a sheet or deposited film made from metals, such as stainless steel, titanium, and tungsten, is the prime candidate.

In addition to the reference sheet, the robustness of probes during repeated contact and withdrawal operations is crucial. The robustness strongly depends on the material and structure of the probes. Kjaer \textit{et al}. and Yamada \textit{et al}. evaluated the robustness of probes through repeated approaches using microscale four probes on samples and demonstrated that stable results were obtained over ten cycles of probe contacts~\cite{Kjaer2008,Yamada2012}. Their results indicate that repeated measurements are possible without significant probe shift provided that the probe material and structure are designed appropriately.

To increase the stability of the contact points for a larger number of microscale probes ($\geq 5$) or even for non-linear multi-probe arrangements, one may use probes aligned with the surface normal, as illustrated in Fig.~\ref{fig:Intro}. In typical multi-probe systems where probes are arranged obliquely~\cite{Shiraki2001,Tanikawa2003,Kim2007,Kjaer2008,Yamada2012}, deflection and horizontal sliding can occur upon engagement, and the resulting shifts in the contact points can become larger as the height variation among the probes increases (generally, as the number of probes increases). 
When employing vertically aligned probes, each probe should possess sufficient elasticity in the vertical direction to ensure reliable contact with the surface without horizontal shift. A potential implementation of this design involves fabricating probe tips with multiple vertically aligned holes through which metallic wires are inserted.

On the other hand, to ensure the contact positions of the reference probes, which directly affect the accuracy of determining target probes via \textit{I--V} measurements, the adoption of an integrated probe-and-sheet system may be an important candidate, where reference probes are fabricated on the reference sheet as illustrated in Fig.~\ref{fig:PossibleSetup}(c). Such an integrated probe-and-sheet system can be mounted on a sample holder and used for determining the target probe positions.
By combining such an integrated probe-and-sheet system with vertically aligned target probes, both accurate determinations of the target probe positions and stable measurements of samples are expected.

The proposed method using \textit{I--V} measurements is particularly useful for determining probe positions that are difficult to access with microscopy, for example, when the probes are hidden behind a probe holder (as illustrated in Fig.~\ref{fig:Intro}) or when the probes are placed inside a radiation shield.
As mentioned previously, the position of the target probes can be determined independently of the resistivity of the reference sheet. Therefore, this method can be applied even at low or high temperatures, where the resistivity of the reference sheet can be changed. 
Although further technical advancements are required to fabricate multi-probe structures at the submicron scale, our proposed probe positioning method is applicable to determining probe positions even at such scales.
Thus, the proposed method simplifies the experimental systems and is useful in extreme environments with temperature variations or where microscopes cannot be installed, thereby providing flexibility in the design of measurement systems and advancing the studies of atomically thin films and surfaces.

Finally, the proposed method provides an approach for determining the multi-probe positions using a reference sheet. The multi-probes of known configuration are applied to measure a target sample of interest, which may exhibit complicated conducting properties such as anisotropic, nonlinear, or mixed 2D and 3D conductions. The measurements are achieved by transferring the target sample before or after the determination of multi-probes with the proposed method. The target sample includes a single-atomic-layer sheet studied in detail by modern scanning tunneling spectroscopy (STM) techniques~\cite{Liu2022}, and also includes a three-dimensional bulk material. Combining STM results with conduction properties obtained using microscale or even submicron-scale multi-probes will provide a deeper insight into microscopic electron transport on films or surfaces.

\begin{figure}
    \centering
    \includegraphics[width=9cm]{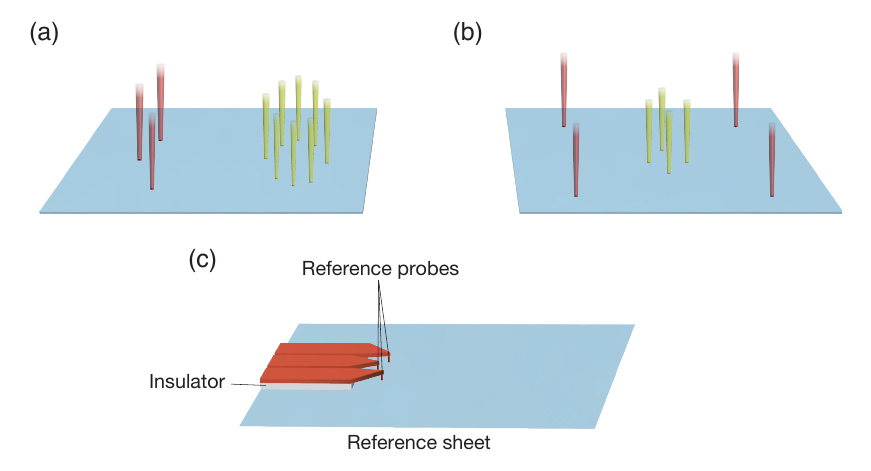}
    \caption{ Examples of reference and target probe configurations for multi-probe measurements. (a) Target probes are placed away from the reference probes. (b) Target probes are surrounded by the reference probes. (c) An integrated probe-and-sheet system, where the reference probes are fabricated on the reference sheet. }
    \label{fig:PossibleSetup}
\end{figure}

\section{CONCLUSIONS}

We proposed a method for determining the positions of multi-probes placed on a two-dimensional sheet. We mathematically proved that the positions of the target probes can be determined via \textit{I--V} measurements. The method was experimentally demonstrated, where the positions of the reference probes and the sheet resistivity were determined from \textit{I--V} measurements with the three reference probes. 
The validity of the method is mathematically guaranteed, and it is applicable even in extreme environments where microscopes are not accessible and are under high and low temperatures.  
Since there are increasing demand for multi-probe measurements at the microscale or submicron-scale order in such environments, the proposed method is expected to advance the research on thin films and surfaces.

\section{ACKNOWLEDGMENTS}
We thank K. Ozawa, K. Sakamoto, and Y. Takeichi for their helpful discussions and acknowledge R. Akiyama and R. Hobara for their advice on designing the setup for the microscale multi-probe measurements. This work was financially supported by Grants-in-Aid for Scientific Research (Grant Nos. 23H01106, 22K17855, and 23K17878) from the Japan Society for the Promotion of Science (JSPS).

%

\appendix

\section{Proof of Theorem~\ref{thm:indepMeas} \label{app:indepMeas}}
Any permutation can be expressed as a composition of transpositions, and therefore, it is sufficient to consider the transpositions: 
\begin{align}
    \tau_{12}(R^{ab}_{cd}) &= R^{ba}_{cd}, \tau_{34}(R^{ab}_{cd}) = R^{ab}_{dc}, \tau_{13}(R^{ab}_{cd}) = R^{cb}_{ad}, \\
    \tau_{24}(R^{ab}_{cd}) &= R^{ad}_{cb}, \tau_{14}(R^{ab}_{cd}) = R^{db}_{ca}, \tau_{23}(R^{ab}_{cd}) = R^{ac}_{bd}. 
\end{align}
From~\eqref{eq:R}, we obtain $\tau_{12} = \tau_{34}$ because
\begin{align}
    \tau_{12}(R^{ab}_{cd}) &= \frac{\rho_{\mathrm{s}}}{2 \pi}\ln\left( \frac{s_{ad}s_{bc}}{s_{ba}s_{cd}} \right) \\
    &= \frac{\rho_{\mathrm{s}}}{2 \pi}\ln\left( \frac{s_{da}s_{cb}}{s_{ab}s_{dc}} \right) = \tau_{34}(R^{ab}_{cd}). 
\end{align}
Similarly, $\tau_{13} = \tau_{24}$ and $\tau_{14} = \tau_{23}$ can also be obtained. Hence, the three transpositions $\tau_{12}$, $\tau_{13}$, and $\tau_{14}$ are sufficient to consider all permutations.

The transposition $\tau_{14}$ corresponds to the swap of $I^+$ and $I^-$ probes, and hence, it only reverses the sign of the resistance, i.e., $\tau_{14} = -e$, where $e$ denotes the identity mapping $e(R^{ab}_{cd}) = R^{ab}_{cd}$.
Definition~\eqref{eq:R} implies that 
\begin{align}
    &\tau_{12}(R^{ab}_{cd}) + \tau_{13}(R^{ab}_{cd}) \\
    =& \frac{\rho_{\mathrm{s}}}{2 \pi}\left\{\ln \left( \frac{s_{ad}s_{bc}}{s_{ba}s_{cd}}\right) + \ln \left( \frac{s_{bd}s_{ca}}{s_{cb}s_{ad}}\right) \right\} \\
    =& \frac{\rho_{\mathrm{s}}}{2 \pi} \ln \left( \frac{s_{bd}s_{ca}}{s_{ba}s_{cd}}\right) = R^{ab}_{cd},
\end{align}
which implies that $\tau_{12} + \tau_{13} = e$ holds.

Consequently, all permutations can be expressed by pairs of the mappings $e$ and $\tau_{12}$ or $e$ and $\tau_{13}$, indicating that all resistance measurements can be computed from $R^{ab}_{cd}$ and $R^{ba}_{cd}\ ( = \tau_{12}(R^{ab}_{cd}))$, or $R^{ab}_{cd}$ and $R^{cb}_{ad}\ ( = \tau_{13}(R^{ab}_{cd}))$.

\section{Proofs of Lemmas~\ref{thm:3probes}, \ref{thm:4probes}, and \ref{thm:5probes} \label{app:proofs}}

In this appendix, we use the same notations as in Section~\ref{sec:Math}.
Before starting the proofs, some additional symbols and relations are introduced to simplify the proofs.
By defining $l_{ab} \coloneqq s_{ab}^2$ and $D^{ab}_{cd} \coloneqq (C^{ab}_{cd})^2$,~\eqref{eq:R2} and~\eqref{eq:cos_rule} can be rewritten as
\begin{equation}
    D^{ab}_{cd}(\rho_{\mathrm{s}}) l_{ab}l_{cd} = l_{bd}l_{ac}, \label{eq:R3}
\end{equation}
and
\begin{equation}
    l_{ab} + l_{ac} - l_{bc} = 2 \bm{v}_{ab}^\top \bm{v}_{ac}. \label{eq:cos_rule2}
\end{equation}
Let $\bm{q}_d \in \mathbb{R}^2$ be a vector defined as
\begin{equation}
    \bm{q}_{d} \coloneqq \begin{bmatrix}
        l_{ab} + l_{ad} - l_{bd} \\
        l_{ac} + l_{ad} - l_{cd}
    \end{bmatrix},
\end{equation}

and $V_{a,bc} \in \mathbb{R}^{2 \times 2}$ be a square matrix defined as 
\begin{equation}
    V_{a,bc} \coloneqq \begin{bmatrix}
        \bm{v}_{ab} & \bm{v}_{ac}
    \end{bmatrix}.
\end{equation}

From~\eqref{eq:cos_rule2}, $\bm{q}_d$ and $V_{a,bc}$ are related as 
\begin{equation}
    \frac{1}{2}\bm{q}_{d} = \begin{bmatrix}
        \bm{v}_{ab}^\top \bm{v}_{ad} \\
        \bm{v}_{ac}^\top \bm{v}_{ad}
    \end{bmatrix} = V_{a,bc}^\top \bm{v}_{ad}, \label{eq:Vbc_q}
\end{equation}
which is frequently used in the following proofs.

Moreover, by substituting~\eqref{eq:R3} into $\bm{q}_d$ and eliminating $l_{bd}$ and $l_{cd}$, $\bm{q}_d$ can be expressed as an affine function of $l_{ad}$: 
\begin{equation}
    \bm{q}_d = \bm{\alpha}_d(\rho_{\mathrm{s}}) l_{ad} + \bm{\beta}, \label{eq:qab}
\end{equation}
where $\bm{\alpha}_d \in \mathbb{R}^2$ and $\bm{\beta} \in \mathbb{R}$ are defined as
\begin{equation}
    \bm{\alpha}_d(\rho_{\mathrm{s}}) \coloneqq \begin{bmatrix}
        1 - \frac{D^{da}_{bc}(\rho_{\mathrm{s}})l_{bc}}{l_{ac}} \\
        1 - \frac{D^{da}_{cb}(\rho_{\mathrm{s}})l_{bc}}{l_{ab}}
    \end{bmatrix}, \quad \bm{\beta} \coloneqq \begin{bmatrix}
        l_{ab} \\
        l_{ac}
    \end{bmatrix}. 
\end{equation}

\subsection{Proof of Lemma~\ref{thm:3probes}}
Let $G_2$ be the $2 \times 2$ matrix defined as
\begin{equation}
    G_2 \coloneqq \begin{bmatrix}
        \bm{v}_{ab}^\top \bm{v}_{ab} & \bm{v}_{ab}^\top \bm{v}_{ac} \\
        \bm{v}_{ac}^\top \bm{v}_{ab} & \bm{v}_{ac}^\top \bm{v}_{ac}
    \end{bmatrix} = V_{a,bc}^\top V_{a,bc}.
\end{equation}
Since the probes $a$, $b$, and $c$ do not align on a straight line, the matrix $V_{a,bc}$ is invertible, and $G_2$ is positive definite.
Hence, we obtain the identity: 
\begin{align}
    0 &= \bm{v}_{ad}^\top \bm{v}_{ad} - \bm{v}_{ad}^\top \bm{v}_{ad}\\
    &= l_{ad} - \bm{v}_{ad}^\top \left(V_{a,bc}V_{a,bc}^{-1} \right) \left(	V_{a,bc}V_{a,bc}^{-1} \right)^\top \bm{v}_{ad} \\
    &= l_{ad} - \bm{v}_{ad}^\top V_{a,bc} \left\{V_{a,bc}^\top V_{a,bc}\right\}^{-1} V_{a,bc}^\top \bm{v}_{ad} \\
    &= l_{ad} - \bm{v}_{ad}^\top V_{a,bc} G_2^{-1} V_{a,bc}^\top \bm{v}_{ad}.
\end{align}
This identity combined with~\eqref{eq:Vbc_q} and~\eqref{eq:qab} provides
\begin{equation}
    l_{ad} - \frac{1}{4} \|\bm{q}_{d}\|^2_{G_2^{-1}} 
    = l_{ad} - \frac{1}{4} \left\| \bm{\alpha}_d l_{ad} + \bm{\beta}\right\|^2_{G_2^{-1}} = 0. \label{eq:tbz_sad}
\end{equation}
where $\|\bm{q}_{d}\|^2_{G_2^{-1}} \coloneqq \bm{q}_{d}^\top G_2^{-1} \bm{q}_{d}$. 
Equation~\eqref{eq:tbz_sad} has at most two positive solutions because it is quadratic with respect to $l_{ad}$.
The candidates for $l_{bd}$ and $l_{cd}$ can be computed from each positive solution of~\eqref{eq:tbz_sad} via~\eqref{eq:R3}.

For a triplet of positive values $(l_{ad}, l_{bd}, l_{cd})$, the position of probe $d$ is uniquely determined as the unique point $\bm{p}_d$ that satisfies $\|\bm{p}_d - \bm{p}_a\| = \sqrt{l_{ad}}$, $\|\bm{p}_d - \bm{p}_b\| = \sqrt{l_{bd}}$, and $\|\bm{p}_d - \bm{p}_c\| = \sqrt{l_{cd}}$. 
Hence, the number of candidates for the position of probe $d$ is at most two.

\subsection{Proof of Lemma~\ref{thm:4probes}}
\label{appsec:4probes}

Let $\bm{L} = [l_{ae}\ l_{be}\ l_{ce}\ l_{de}]^\top \in \mathbb{R}^4$ be the vector of unknown variables, i.e., the squared distances between probe $e$ and probes $a$, $b$, $c$, and $d$.
Equation~\eqref{eq:R3} for quadroplets $(e, a, b, c)$, $(e, a, c, d)$, and $(e, a, d, b)$ can be rewritten as
\begin{equation}
    A_r \bm{L} = 0, \label{eq:Ar}
\end{equation}
where $A_r \in \mathbb{R}^{3 \times 4}$ is a matrix defined as
\begin{equation}
    A_r := \begin{bmatrix}
        D^{ea}_{bc}l_{bc} & -l_{ac} & 0 & 0 \\
        D^{ea}_{cd}l_{cd} & 0 & -l_{ad} & 0 \\
        D^{ea}_{db}l_{bd} & 0 & 0 & -l_{ab}
    \end{bmatrix}. 
\end{equation}

It will be shown that \eqref{eq:Ar} has a unique solution under a certain condition. 
To this end, we first show that a component of $\bm{L}$ is uniquely determined by the other components from the geometric condition~\eqref{eq:cos_rule2}. 
From~\eqref{eq:Vbc_q} and the assumption that probes $a$, $b$, and $c$ do not align on a straight line, we obtain
\begin{equation}
\bm{v}_{ae} = \frac{1}{2} V_{a,bc}^{-\top}\bm{q}_{e},
\end{equation}
where $V_{a,bc}^{-\top} = (V_{a,bc}^\top)^{-1}$. 
By substituting this into~\eqref{eq:cos_rule2}, we obtain
\begin{equation}
    l_{ad} + l_{ae} - l_{de} = 2 \bm{v}_{ad}^\top \bm{v}_{ae} = \bm{v}_{ad}^\top V_{a,bc}^{-\top}\bm{q}_e. \label{eq:s34}
\end{equation}
The row vector $\bm{v}_{ad}^\top V_{a,bc}^{-\top}$ can be rewritten as
\begin{align}
    \bm{v}_{ad}^\top V_{a,bc}^{-\top} &= \bm{v}_{ad}^\top \begin{bmatrix}
        \bm{v}_{ac, y} &  -\bm{v}_{ab, y}\\ 
        -\bm{v}_{ac, x} & \bm{v}_{ab, x}
    \end{bmatrix}/|V_{a,bc}| \\
    &= -\frac{1}{|V_{a,bc}|} \begin{bmatrix}
        |V_{a,cd}| & |V_{a,db}|
    \end{bmatrix},
\end{align}
where subscripts $x$ and $y$ denote the $x$- and $y$-components of the vectors, respectively.
Direct calculations show that $|V_{a,bc}| + |V_{a,cd}| + |V_{a,db}| = |V_{b, cd}|$, which is nonzero as long as probes $b$, $c$, and $d$ do not align on a straight line.
Consequently,~\eqref{eq:s34} can be rearranged as
\begin{multline}
    l_{ae} = \frac{|V_{a,cd}|}{|V_{b, cd}|}l_{be} + \frac{|V_{a,db}|}{|V_{b, cd}|}l_{ce} 
    + \frac{|V_{a,bc}|}{|V_{b, cd}|}l_{de} \\
    - \frac{ l_{ab}|V_{a,cd}| + l_{ac}|V_{a,db}| + l_{ad}|V_{a,bc}|}{|V_{b, cd}|}. \label{eq:s04}
\end{multline}

From~\eqref{eq:s04}, $l_{ae}$ can be represented as a linear function of $l_{be}$, $l_{ce}$, and $l_{de}$.  
Hence, by letting $\tilde{\bm{L}} = [l_{be}\ l_{ce}\ l_{de}]^\top$, $\bm{L}$ can be expressed using $\tilde{\bm{L}}$ as
\begin{equation}
    \bm{L} = A_g\tilde{\bm{L}} + \bm{b}_g, \label{eq:Ag}
\end{equation}
where
\begin{align}
    A_g &\coloneqq \frac{1}{|V_{b, cd}|}\begin{bmatrix}
        |V_{a,cd}| & |V_{a,db}| & |V_{a,bc}| \\
        |V_{b, cd}| & 0 & 0 \\
        0 & |V_{b, cd}| & 0 \\
        0 & 0 & |V_{b, cd}| \\
    \end{bmatrix}, \\
    \bm{b}_g &\coloneqq - \frac{1}{|V_{b, cd}|}\begin{bmatrix}
        l_{ab}|V_{a,cd}| + l_{ac}|V_{a,db}| + l_{ad}|V_{a,bc}| \\ 0 \\ 0 \\ 0 
    \end{bmatrix}. 
\end{align}

By substituting~\eqref{eq:Ag} into~\eqref{eq:Ar}, we obtain a linear equation with respect to $\tilde{\bm{L}}$ as
\begin{equation}
    A_rA_g \tilde{\bm{L}} = -A_r \bm{b}_g, \label{eq:linsys}
\end{equation}
which has a unique solution if and only if
\begin{multline}
        |A_rA_g| = \Bigl( l_{ac}l_{ad}(D^{ea}_{bd}l_{bd} - l_{ab})|V_{a,bc}| \\
        + l_{ab}l_{ad}(D^{ea}_{bc}l_{bc} - l_{ac})|V_{a,cd}| \\
        + l_{ab}l_{ac}(D^{ea}_{cd}l_{cd} - l_{ad})|V_{a,db}|\Bigr)/|V_{b, cd}| \neq 0, \label{eq:4probes_condition}
\end{multline}
holds. 
Finally, if $\tilde{\bm{L}} = [l_{be}\ l_{ce}\ l_{de}]^\top$ is uniquely determined, the position $\bm{p}_e$ can be uniquely determined because probes $b$, $c$, and $d$ do not align on a straight line.

\subsection{Proof of Lemma~\ref{thm:5probes}}
To prove Lemma~\ref{thm:5probes}, we first derive a condition for $\rho_{\mathrm{s}}$ that ensures the existence of probe positions satisfying~\eqref{eq:R2}. 
This can be achieved by introducing the notion of the resultant of two polynomials, providing a condition for these polynomials to have common roots.
\begin{definition}[\cite{Cox2005}]
    For two indeterminates $(X, Y)$, let $f$ and $g$ be polynomials in $X$ and $Y$, that is, 
    \begin{align}
        f(X, Y) \coloneqq a_0 + a_1X + \cdots + a_nX^n, \\
        g(X, Y) \coloneqq b_0 + b_1X + \cdots + b_mX^m,
    \end{align}
    where $a_i\ (i=0,\dots,n)$ and $b_j\ (j=0,\dots,m)$ are polynoimals in $Y$ with real coefficients.
    Then, the resultant of $f$ and $g$ with respect to $X$ is defined as the determinant of the Sylvester matrix: 
    \begin{align}
        \mathrm{Res}(f, g, X) \coloneqq &\left| 
        \begin{bmatrix}
        a_0 & 0 & \cdots & 0 & b_0 & 0 & \cdots & 0 \\
        a_1 & a_0 & \cdots & 0 & b_1 & b_0 & \cdots & 0 \\
        \vdots & a_1 & \ddots & \vdots & \vdots & b_1 & \ddots & \vdots \\
        a_n & \vdots & \ddots & a_0 & b_m & \vdots & \ddots & b_0 \\
        0 & a_n & \cdots & a_1 & 0 & b_m & \cdots & b_1 \\
        \vdots & \vdots & \ddots & \vdots & \vdots & \vdots & \ddots & \vdots \\
        0 & 0 & \cdots & a_n & 0 & 0 & \cdots & b_m
        \end{bmatrix}
        \right|, \\[-1em]
        &\quad \underbrace{\qquad\qquad\quad\quad}_{m} \quad \underbrace{\qquad\qquad\quad\quad}_{n}
    \end{align}
    which is a polynomial in $Y$ with real coefficients.
\end{definition}

The relationship between the resultant and common roots of two polynomials is well known and can be summarized as the following lemma.
\begin{lemma}[Elimination property of resultants~\cite{Cox2005}] \label{lem:res1}
    Let $f(X, Y)$ and $g(X, Y)$ be polynomials in $X$ and $Y$ with real coefficients.
    Suppose that $f$ and $g$ have a common zero $(\hat{X}, \hat{Y}) \in \mathbb{R}^2$. 
    Then, $\mathrm{Res}(f, g, X)$ vanishes at $Y = \hat{Y}$. 
    In particular, if $f$ and $g$ do not depend on $Y$ and have a common zero, then $\mathrm{Res}(f, g, X) = 0$.
\end{lemma}

From~\eqref{eq:tbz_sad}, we have 
\begin{align}
        4l_{ad} - \bm{q}_d^\top G_2^{-1} \bm{q}_d = 0 \label{eq:fd}\\
        4l_{ae} - \bm{q}_e^\top G_2^{-1} \bm{q}_e = 0. \label{eq:fe}
\end{align}
On the other hand, from\eqref{eq:Vbc_q} and~\eqref{eq:s34}, we have
\begin{align}
    l_{ad} + l_{ae} - l_{de} &= \bm{v}_{ad}^\top \left( V_{a,bc} V_{a,bc}^{-1} \right)V_{a,bc}^{-\top}\bm{q}_e \\
    &= (V_{a,bc}^\top \bm{v}_{ad})^\top \{V_{a,bc}^\top V_{a,bc}\}^{-1}\bm{q}_e \\
    &= \frac{1}{2}\bm{q}_d G_2^{-1} \bm{q}_e. \label{eq:pre_fde}
\end{align}
From~\eqref{eq:fd} and~\eqref{eq:fe}, both $l_{ad}$ and $l_{ae}$ are expressed as quadratic forms of $\bm{q}_d$ and $\bm{q}_e$, respectively. 
Substituting the quadratic forms into~\eqref{eq:pre_fde}, we obtain 
\begin{equation}
   4 l_{de} - (\bm{q}_d - \bm{q}_e)^\top G_2^{-1} (\bm{q}_d - \bm{q}_e) = 0. \label{eq:fde}
\end{equation}
By combining~\eqref{eq:R3} for quadroplets $(b, d, c, e)$, $(d, a, b, c)$, and $(e, a, c, b)$, we obtain
\begin{equation}
    l_{de} = \frac{D^{bd}_{ce}(\rho_{\mathrm{s}})D^{da}_{bc}(\rho_{\mathrm{s}})D^{ea}_{cb}(\rho_{\mathrm{s}})l_{bc}}{l_{ab}l_{ac}}l_{ad}l_{ae}. 
\end{equation}
Consequently, we obtain the following equations. 
\begin{align}
    &f_{d, \rho_{\mathrm{s}}}(l_{ad}) \coloneqq 4l_{ad} - \left\| \bm{\alpha}_d(\rho_{\mathrm{s}})l_{ad} + \bm{\beta} \right\|^2_{G_2^{-1}} = 0, \label{eq:fd2} \\
    &f_{e, \rho_{\mathrm{s}}}(l_{ae}) \coloneqq 4l_{ae} - \left\| \bm{\alpha}_e(\rho_{\mathrm{s}})l_{ae} + \bm{\beta} \right\|^2_{G_2^{-1}} = 0, \label{eq:fe2} \\
    &\begin{multlined}
        f_{de, \rho_{\mathrm{s}}}(l_{ad}, l_{ae}) \coloneqq 
        \left\| \bm{\alpha}_d(\rho_{\mathrm{s}})l_{ad} - \bm{\alpha}_e(\rho_{\mathrm{s}})l_{ae} \right\|^2_{G_2^{-1}} \\
        - 4\frac{D^{bd}_{ce}(\rho_{\mathrm{s}})D^{da}_{bc}(\rho_{\mathrm{s}})D^{ea}_{cb}(\rho_{\mathrm{s}})l_{bc}}{l_{ab}l_{ac}}l_{ad}l_{ae} = 0, \label{eq:fde2}
    \end{multlined} 
\end{align}

From Lemma~\ref{lem:res1}, if \eqref{eq:fd2} and~\eqref{eq:fde2} have a common root, the resultant $\mathrm{Res}(f_d, f_{de}, l_{ad})$, which is a polynomial in $l_{ae}$, vanishes at the common root. 
Moreover, Lemma~\ref{lem:res1} ensures that if~\eqref{eq:fe2} and $\mathrm{Res}(f_d, f_{de}, l_{ad})$ have a common root, the resultant $\mathrm{Res}(f_e, \mathrm{Res}(f_d, f_{de}, l_{ad}), l_{ae})$, which is a polynomial equation in $D^{da}_{bc}$, $D^{da}_{cb}$, $D^{ea}_{bc}$, $D^{ea}_{cb}$, and $D^{bd}_{ce}$ depending on a single parameter $\rho_{\mathrm{s}}$, must vanish at the common root. 
Hence, the true $\rho_{\mathrm{s}}$ necessarily satisfies the nonlinear equation 
\begin{equation}
    \mathrm{Res}(f_e, \mathrm{Res}(f_d, f_{de}, l_{ad}), l_{ae}) = 0. \label{eq:resres}
\end{equation}
In other words, all solutions of~\eqref{eq:resres} can be regarded as candidates for $\rho_{\mathrm{s}}$.
Therefore, we will show that the number of solutions to~\eqref{eq:resres} is at most finite, which implies that only a finite number of triplets $(l_{ad}, l_{ae}, \rho_{\mathrm{s}})$ satisfy~\eqref{eq:fd2}--\eqref{eq:fde2}.

By defining $\sigma_{\mathrm{s}} \coloneqq 1/\rho_{\mathrm{s}}$ and $\gamma^{ab}_{cd} \coloneqq 4 \pi R^{ab}_{cd}$, $D^{ab}_{cd}$ can be expressed as $D^{ab}_{cd} = \exp(\gamma^{ab}_{cd} \sigma_{\mathrm{s}})$.
For simplicity, $D^{da}_{bc}$, $D^{da}_{cb}$, $D^{ea}_{bc}$, $D^{ea}_{cb}$, and $D^{bd}_{ce}$ are denoted by $Z_1$, $Z_2$, $Z_3$, $Z_4$, and $Z_5$, respectively, and are expressed as
\begin{equation}
    Z_i = \exp(\gamma_i \sigma_{\mathrm{s}}), \quad i = 1, \dots, 5. \label{eq:Z}
\end{equation}
The resultant can then be computed from~\eqref{eq:fd2}--\eqref{eq:fde2} as the following polynomial of the total degree 20 in $Z_i\ (i=1,\dots,5)$:
\begin{align}
    &\mathrm{Res}(f_e, \mathrm{Res}(f_d, f_{de}, l_{ad}), l_{ae}) \\
    &= \frac{256l_{bc}^{16}}{|G_2|^8 l_{ac}^4} Z_1^8 Z_3^4 Z_4^4 Z_5^4 + \cdots \\
    &= a_{[8\ 0\ 4\ 4\ 4]^\top} Z_1^8Z_2^0Z_3^4Z_4^4Z_5^4 + \cdots \\
    &= \sum_{k \in K} a_k Z^k, 
\end{align}
where $k = [k_1\ \cdots\ k_5]^\top \in \bm{N}^5$ is a multi-index that specifies a monomial $Z^k \coloneqq Z_1^{k_1}Z_2^{k_2} \cdots Z_5^{k_5}$, $a_k$ is the coefficient of the monomial $Z^k$, $K \coloneqq \{k \in \bm{N}^5 \mid a_k \neq 0\}$, and $\gamma \coloneqq [\gamma_1\ \cdots\ \gamma_5]^\top$.
By substituting~\eqref{eq:Z} into the left-hand side of~\eqref{eq:resres}, we obtain
\begin{equation}
    g(\sigma_{\mathrm{s}}) \coloneqq \sum_{k \in K} a_k e^{k^\top \gamma \sigma_{\mathrm{s}}}. 
\end{equation}
We will show that $g(\sigma_{\mathrm{s}})$ has only a finite number of zeros, which completes the proof of Lemma~\ref{thm:5probes}.

First, we show that there exists $\hat{\sigma}_s > 0$ such that $g(\sigma_{\mathrm{s}}) \neq 0$ for all $\sigma_{\mathrm{s}} > \hat{\sigma}_s$. 
For $k_{\mathrm{max}} \coloneqq \argmax_{k \in K}(k^\top \gamma)$, we obtain
\begin{equation}
    g(\sigma_{\mathrm{s}})e^{-k_{\mathrm{max}}^\top \gamma \sigma_{\mathrm{s}}} = 
    a_{k_{\mathrm{max}}} + \sum_{k \in K \setminus \{k_{\mathrm{max}}\}} a_k e^{(k - k_{\mathrm{max}})^\top \gamma \sigma_{\mathrm{s}}}.
\end{equation}  
This indicates that $g(\sigma_{\mathrm{s}})\exp(-k_{\mathrm{max}}^\top \gamma \sigma_{\mathrm{s}}) \to a_{k_{\mathrm{max}}}$ as $\sigma_{\mathrm{s}} \to \infty$. 
More specifically, for any positive real number $\varepsilon > 0$, there exists $\hat{\sigma}_s$ such that for all $\sigma_{\mathrm{s}} > \hat{\sigma}_s$, 
$|a_{k_{\mathrm{max}}} - g(\sigma_{\mathrm{s}})\exp(-k_{\mathrm{max}}^\top \gamma \sigma_{\mathrm{s}})| < \varepsilon$. 
Therefore, by choosing $\varepsilon$ as $|a_{k_{\mathrm{max}}}|/2$, we obtain 
\begin{align}
    \frac{|a_{k_{\mathrm{max}}}|}{2} &> \left| a_{k_{\mathrm{max}}} - g(\sigma_{\mathrm{s}})\exp(-k_{\mathrm{max}}^\top \gamma \sigma_{\mathrm{s}}) \right| \\
    &\geq \left| a_{k_{\mathrm{max}}} \right| - \frac{\left| g(\sigma_{\mathrm{s}}) \right|}{\exp(k_{\mathrm{max}}^\top \gamma \sigma_{\mathrm{s}})} 
\end{align}
which indicates 
\begin{equation}
    |g(\sigma_{\mathrm{s}})| > \frac{|a_{k_{\mathrm{max}}}|}{2} \exp(k_{\mathrm{max}}^\top \gamma \sigma_{\mathrm{s}}) > 0 
\end{equation}
for all $\sigma_{\mathrm{s}} > \hat{\sigma}_s$.

Now, all zeros of $g(\sigma_{\mathrm{s}})$ must be contained in $\sigma_{\mathrm{s}} \coloneqq (0, \hat{\sigma}_s)$.
If $g(\sigma_{\mathrm{s}})$ has infinitely many zeros in $\sigma_{\mathrm{s}}$, the Bolzano--Weierstrass theorem~\cite{Ahlfors1953} implies that the zero set of $g(\sigma_{\mathrm{s}})$ contains an accumulation point in the compact set $\overline{\Sigma}_s$.
By the identity theorem of analytic functions~\cite{Ahlfors1953}, it implies that $g(\sigma_{\mathrm{s}}) = 0\ (\sigma_{\mathrm{s}} \in \bm{C})$, which contradicts $a_k \neq 0\ (k \in K)$.
Therefore, $g(\sigma_{\mathrm{s}})$ has only a finite number of zeros in $\sigma_{\mathrm{s}}$, which completes the proof of Lemma~\ref{thm:5probes}.

\end{document}